\documentclass[twocolumn]{IEEEtran}

\usepackage{amsmath,epsfig,amssymb,verbatim,amsopn,cite,multirow}
\usepackage[usenames,dvipsnames]{color}
\usepackage{algorithm,algorithmic}
\usepackage{amsfonts}
\usepackage{amssymb}
\usepackage{epsfig}
\usepackage{epstopdf}
\usepackage{graphicx}
\usepackage{footnote}

\usepackage[margin=15mm,top=18mm,bottom=26mm]{geometry}

\usepackage[nodisplayskipstretch]{setspace}
\newtheorem{proposition}{Proposition}
\newtheorem{proof}{Proof}
\usepackage{amsmath,epsfig,amssymb,verbatim,amsopn,cite,subfigure,multirow}
\usepackage{balance}
\usepackage{cite}

\usepackage{algorithm,algorithmic}

\usepackage{MnSymbol}
\begin{document}

\author{
Mustafa S. Abbas,  Zahra Mobini, Hien Quoc Ngo, and Michail Matthaiou\\
\small{
Centre for Wireless Innovation (CWI), Queen's University Belfast, U.K.\\
Email:\{malkhadhrawee01, zahra.mobini, hien.ngo, m.matthaiou\}@qub.ac.uk}}\normalsize

\title{Unicast-Multicast Cell-Free Massive MIMO: Gradient-Based Resource Allocation}

\maketitle
\begin{abstract}
We consider a cell-free massive multiple-input multiple-output (CF-mMIMO) system with joint unicast and multi-group multicast transmissions. We derive exact closed-form expressions for the downlink achievable spectral efficiency (SE) of both unicast and multicast users.
Based on these expressions, we formulate a joint optimization problem of access point (AP) selection and power control  subject to  quality of service (QoS) requirements of all unicast and multicast users and per-AP maximum  transmit power constraint. The challenging
formulated problem is transformed into a tractable form
and a novel accelerated projected gradient (APG)-based algorithm is developed to solve the optimization problem. Simulation results show that our joint optimization strategy enhances notably the sum SE (SSE) (up to {$58\%$}) compared to  baseline schemes, while maintaining  low complexity.   

\let\thefootnote\relax\footnotetext{The authors  are with the Centre for Wireless Innovation (CWI), Queen's University Belfast, BT3 9DT Belfast, U.K. email:\{malkhadhrawee01, zahra.mobini, hien.ngo, m.matthaiou\}@qub.ac.uk.

This work is a contribution by Project REASON, a UK Government funded project under the Future Open Networks Research Challenge (FONRC) sponsored by the Department of Science Innovation and Technology (DSIT).
The work of M.S. Abbas, Z. Mobini and H. Q. Ngo was supported by the U.K. Research and Innovation Future Leaders Fellowships under Grant MR/X010635/1, and a research grant from the Department for the Economy Northern Ireland under the US-Ireland R\&D Partnership Programme. The work of M. Matthaiou has received funding from the European Research Council (ERC) under the European Union’s Horizon 2020 research and innovation programme (grant agreement No. 101001331).}
\end{abstract}
\vspace{-1em}
\section{Introduction}
CF-mMIMO is a highly promising technology for next-generation wireless systems \cite{ngo2024ultra,Matthaiou:COMMag:2021}. It involves deploying a large number of  APs  across a broad geographic area, which operate coherently within the same time-frequency resources using time division duplex (TDD) to serve multiple users without the need for cells or cell boundaries.
CF-mMIMO offers many advantages, such as high connectivity along with substantial SE and energy efficiency \cite{ngo2024ultra}. Thus, it has recently attracted a lot of research attention.  In \cite{Ngo:IEEE_G_Net:2018}, the authors considered  a unicast CF-mMIMO with multiple antenna APs and proposed a successive convex approximation (SCA) method to optimize the total energy efficiency. Moreover, AP selection was proposed to improve the performance. In \cite{Farooq:TCOM:2021}, the authors proposed an APG solution to maximize several system-wide utility functions and showed that the APG is much less complex and more memory efficient as compared to the SCA method used in \cite{Ngo:IEEE_G_Net:2018}. 
For multicast systems, \cite{Farooq:EUSIPCO:2021} provided a comparison between the bisection and the APG method to solve a max-min fairness (MMF) power allocation problem.

All of the aforementioned studies have only considered either a unicast or multicast system, while in many practical scenarios and for future massive access, wireless systems include both unicast and multicast transmissions. Thus, it is of practical importance to have efficient resource allocation schemes in  joint unicast and multicast transmission CF-mMIMO systems. However, since joint unicast and multicast transmissions introduce additional constraints to the related multi-objective optimization problems, there has been little research done on this topic. The authors in \cite{Sadeghi:WCOM:2018,Li:TVT:2022,Tan:ISNCC:2020} considered joint unicast and multicast systems with different optimization methods. In \cite{Sadeghi:WCOM:2018}, the authors exploited a multi-objective optimization problem (MOOP) and Pareto boundary to optimize the  SSE for unicast and MMF for multicast users in single-cell cellular massive MIMO. In \cite{Li:TVT:2022}, the authors considered a joint unicast and multicast system in CF-mMIMO. The problem of MOOP was solved based on deep learning and the non-dominated sorting genetic algorithm II (NSGA-II). The results showed that deep learning was better in terms of elapsed time, while NSGA-II was better in terms of sum achievable rate. The work in \cite{Tan:ISNCC:2020} discussed the energy efficiency of layered-division multiplexing for joint unicast and multicast transmission. The authors solved the optimization problem using the SCA and Dinkelbach's method. 

On a parallel note, user association/AP selection in CF-mMIMO has attracted a lot of research interest thanks to its potential to reduce the fronthaul/backhaul signaling load and simplify the AP complexity, thereby making CF-mMIMO more implementable \cite{ hao2024joint,Yassenwcnc}. Despite its importance, user association/AP selection for massive access in CF-mMIMO networks, especially for joint unicast and multicast systems, has not been well investigated.
Inspired by the above observation, this paper studies a joint unicast and multicast CF-mMIMO system. To enhance its performance, we propose a novel power control and user association scheme based on the APG algorithm. The APG is selected due to its superior memory efficiency and lower complexity compared to other methods. The paper's main contributions are as follows:
\begin {itemize}
\item We derive a closed-form SE expression for joint unicast and multicast CF-mMIMO using the use-and-then-forget bounding technique. It is worth noting that in \cite{Li:TVT:2022}, a closed-form expression for the SE was derived. However, that was an approximate result under the assumption that the number of antennas tends to infinity.

\item We formulate a joint optimization problem of AP selection and power control for joint unicast-multicast CF-mMIMO under AP transmit power constraint and individual QoS requirements for both the unicast and multicast users. A new APG-based algorithm is developed to solve the above optimization problem.

\item Numerical results show that the proposed APG-based  optimization solution  improves significantly the SSE of the joint unicast and multicast CF-mMIMO systems.
\end {itemize}

\textit{Notation:} The symbols $(\cdot)^T$, $(\cdot)^\ast$, and $(\cdot)^H$ represent the transpose, conjugate, and conjugate transpose, respectively. The symbols  
$\mathbf{I}_n$ and  $\mathbb{E}\{\cdot\}$ stand for the $n\times n$ identity matrix, and the statistical expectation, respectively. 
Finally, a  circular symmetric complex Gaussian variable having variance $\sigma^2$ is denoted by $\mathcal{CN}(0,\sigma^2)$. 

\vspace{-0.5em}
\section{System Model}

We consider a CF-mMIMO system with joint unicast and multi-group multicast transmissions. The CF-mMIMO  system consists of $N$ APs, each equipped with $L$ antennas, that serve simultaneously $U$ unicast users and $M$ multicast groups, where each group includes $K_m$ users. The sets of $N$ APs, $U$ unicast users, $M$ multicast groups, and $K_m$ users in the $m$-th multicast group are denoted by $\mathcal{N}$, $\mathcal{U}$, $\mathcal{M}$, and $\mathcal{K}_m$, respectively.

The  channel vector  between the $u$-th unicast user, $u \in \mathcal{U}$, and the $n$-th AP, $n \in \mathcal{N}$, is $\mathbf c_{n,u} = \beta_{n,u}^{1/2} \mathbf h_{n,u}\in \mathbb{C}^{L\times1}$. Moreover, the channel between the $k_m$-th multicast user, $k_m \in \mathcal{K}_m$, of the $m$-th multicast group, $m \in \mathcal{M}$, and the $n$-th AP is $\mathbf t_{n,m,k} = \lambda_{n,m,k}^{1/2} \mathbf h_{n,m,k}\in \mathbb{C}^{L\times1}$, where $\beta_{n,u}$ and $\lambda_{n,m,k}$ represent the large-scale fading coefficients, while 
$\mathbf h_{n,u} \sim \mathcal{CN} (\mathbf 0,\mathbf I_{L})$ and $\mathbf h_{n,m,k} \sim \mathcal{CN} (\mathbf 0,\mathbf I_{L})$ are the small-scale fading vectors. 

\vspace{-0.5em}
\subsection{Uplink Training}
The system is assumed to work under the reciprocity-based TDD protocol, where the channels remain unchanged during a coherence interval $T$. The channel state information (CSI) is is obtained through  uplink training. We assume that the pilots dedicated to the unicast users are orthogonal. Since in practice the length of the coherence  interval is limited, we assign a shared pilot to all the users in each multicast group~\cite{yang:SPAWC:2013}. Therefore, the CF-mMIMO system requires $U + M$ orthogonal pilots. We use $ \boldsymbol{\phi}_{u}\in\mathbb{C}^{\tau  \times 1}$,  $\|\boldsymbol{\phi}_{u} \|^2 =1$ to denote the pilot sequence sent by the $u$-th unicast user, and  $ \boldsymbol{\varphi}_{m}\in\mathbb{C}^{\tau  \times 1}$,  $\|\boldsymbol{\varphi}_{m} \|^2 =1$, to denote the pilot assigned to   all the multicast users in the $m$-th multicast group, while  the pilot length $\tau$ satisfies the condition $U+M \leq \tau \leq T$.
 The received signal at the $n$-th AP during the uplink training is
\vspace{0.1 em}
\begin{align} \label{eq:ynp}
    \mathbf Y_{n,p}=&\sqrt{\tau p_{\text{ul}}}\sum_{u \in \mathcal{U}}  \mathbf c_{n,u} \boldsymbol{\phi}_{u}^H\nonumber\\
    & + \sqrt{\tau p_{\text{ul}}} \sum_{m \in \mathcal{M}}  \sum_{k   \in \mathcal{K}_m}  \mathbf t_{n,m,k} \boldsymbol{\varphi}_{m}^H + \mathbf W_{n,p},
\end{align}
where $\mathbf W_{n,p}$ is the additive noise while $p_{\text{ul}}$ is the uplink transmit power.
The received signal  is then projected onto the  pilot sequence associated with the $u$-th unicast user, we obtain
\vspace{0.1 em}
\begin{equation} \label{eq:ynpu}
    \mathbf{\Check{y}}_{n,p,u} = \mathbf Y_{n,p} \boldsymbol{\phi}_{u} = \sqrt{\tau p_{\text{ul}}}  \mathbf c_{n,u}  + \mathbf w_{n,p}.
\end{equation}
The minimum mean-square error (MMSE) estimate of $\mathbf c_{n,u}$ is 
$\mathbf {\hat{c}}_{n,u} = \frac{\sqrt{\tau p_{\text{ul}}}~ \beta_{n,u}}{\tau p_{\text{ul}}~ \beta_{n,u} + 1} \mathbf{\Check{y}}_{n,p,u}.$
The variance of the estimated channel 
$\mathbf {\hat{c}}_{n,u}$ can be obtained as $\mathbf\gamma_{n,u} = \frac{\tau p_{\text{ul}}~ \beta_{n,u}^2}{\tau p_{\text{ul}}~ \beta_{n,u} + 1}
$. Also, the MMSE estimate of  $\mathbf t_{n,m,k}$ is 
$\mathbf {\hat{t}}_{n,m,k}
= \frac{\sqrt{\tau p_{\text{ul}}}~ \lambda_{n,m,k}}{\tau p_{\text{ul}}~ \sum_{t \in \mathcal{K}_m} \lambda_{n,m,t} + 1} \mathbf{\Check{y}}_{n,p,m},$
where $\mathbf{\Check{y}}_{n,p,m} = \mathbf Y_{n,p}\boldsymbol{\varphi}_{m}$, while the variance of $\mathbf {\hat{t}}_{n,m,k} $ is $\mathbf\xi_{n,m,k} = \frac{\tau p_{\text{ul}}~ \lambda_{n,m,k}^2}{\tau p_{\text{ul}}~ \sum_{t \in \mathcal{K}_m} \lambda_{n,m,t} + 1}$. 
Denote by \cite{yang:SPAWC:2013}
\begin{equation} \label{eq:tnmh}
\mathbf {\hat{t}}_{n,m} = \sum_{k \in \mathcal{K}_m} \mathbf{\hat{t}}_{n,m,k} = \frac{\sqrt{\tau p_{\text{ul}}}~ \sum_{k \in \mathcal{K}_m} \lambda_{n,m,k}}{\tau p_{\text{ul}}~ \sum_{k \in \mathcal{K}_m} \lambda_{n,m,k} + 1} \mathbf{\Check{y}}_{n,p,m},
\end{equation}
which can be regarded as the channel estimation of multicast users. The mean square of $\mathbf {\hat{t}}_{n,m}$ is $\mathbf\zeta_{n,m}= \frac{(\sqrt{\tau p_{\text{ul}}}~ \sum_{k \in \mathcal{K}_m}\lambda_{n,m,k})^2}{\tau p_{\text{ul}}~ \sum_{k \in \mathcal{K}_m} \lambda_{n,m,k} + 1}$.
\begin{figure*}
\begin{equation} \label{eq:sinru}
\mathrm{SINR}_u = \frac{( \sqrt{p_{\text{dl}}}~L~\sum_{n \in \mathcal{N}} a_{n,u} \sqrt{\eta_{n,u}} \mathbf\gamma_{n,u})^2}{ p_{\text{dl}}~ L\sum_{u' \in \mathcal {U}} \sum_{n \in \mathcal{N}} a_{n,u'} \eta_{n,u'} \beta_{n,u} \gamma_{n,u'} + p_{\text{dl}} ~ L\sum_{m \in \mathcal{M}} \sum_{n \in \mathcal{N}} \bar{a}_{n,m} \bar{\eta}_{n,m}   \beta_{n,u}  \zeta_{n,m} +1}~\tag{10}
\end{equation}
\begin{equation} \label{eq:sinrmk}
\mathrm{SINR}_{m,k} = \frac{( \sqrt{p_{\text{dl}}} ~ L\sum_{n \in \mathcal{N}} \bar{a}_{n,m} \sqrt{\bar{\eta}_{n,m}} ~ \mathbf\xi_{n,m,k} )^2}{ p_{\text{dl}}  ~L~\sum_{m' \in \mathcal{M}} \sum_{n \in \mathcal{N}} \bar{a}_{n,m'} \bar{\eta}_{n,m'}\lambda_{n,m,k}\zeta_{n,m'} + p_{\text{dl}} ~L \sum_{u \in \mathcal {U}} \sum_{n \in \mathcal{N}} a_{n,u} \eta_{n,u} \lambda_{n,m,k} \gamma_{n,u} +1}~\tag{11} 
\end{equation}
\hrulefill
	\vspace{-4mm}
\end{figure*}
\vspace{-0.5em}
\subsection{Downlink  Data Transmission and SE Analysis}
Using the local channel estimates, the APs carry out maximum-ratio (MR) precoding for the signals sent to both unicast and multicast users.  We select MR precoding due to its minimal computational complexity and its potential for distributed implementation \cite{Hien:JWCOM:2017,Li:TVT:2022}. In addition, we define the binary variables $a_{n,u}$ and $\bar{a}_{n,m}$ to indicate the AP-user association  for the $u$-th unicast user and the $m$-th multicast group, respectively.  In particular,  we define
\begin{equation}\label{eq:anu}
a_{n,u} \triangleq
\begin{cases}
  1, & \text{if $u$-th unicast user is served by $n$-th AP,}\\
  0, & \mbox{othewise},
\end{cases} 
  \end{equation}
  \begin{equation}\label{eq:anm}
\bar{a}_{n,m} \triangleq
\begin{cases}
  1, & \text{if $m$-th multicast group is served by $n$-th AP,}\\
  0, & \mbox{othewise}.
\end{cases} 
  \end{equation}
Let $q_u$ and $\bar{q}_m$ be the symbols assigned to the $u$-th unicast user and the $m$-th multicast group user, respectively, where $\mathbb{E}\{|q_u|^2\}=\mathbb{E}\{|\bar{q}_m|^2\}=1$. Then, the transmitted signal from the $n$-th AP will be
\vspace{-0.1em}
\begin{align}\label{eq:xn}
\boldsymbol{x}_{n} &=\sqrt{p_{\text{dl}}} \sum_{u \in \mathcal {U}} a_{n,u} \sqrt{\eta_{n,u}} \mathbf{\hat{c}^*}_{n,u} q_u  \nonumber
\\
&+ \sqrt{p_{\text{dl}}} \sum_{m \in \mathcal{M}} \bar{a}_{n,m} \sqrt{\bar{\eta}_{n,m}} \mathbf{\hat{t}^*}_{n,m} \bar{q}_m,
\end{align}
where $p_{\text{dl}}$ is the normalized transmit power at each AP, while $\eta_{n,u}$ and $\bar{\eta}_{n,m}$ are the power control coefficients allocated to the $u$-th unicast user and the $m$-th multicast group, respectively. The power allocation coefficients  must satisfy the constraint    $\mathbb{E}\big\{| |\boldsymbol{x_n}||^2\big\}  \leq p_{\text{dl}}$, which can be rewritten as
\vspace{0.2 em}
\begin{equation}\label{eq:Exn}
     L \sum_{u \in \mathcal {U}} a_{n,u}^2 \eta_{n,u} \mathbf\gamma_{n,u}  + L \sum_{m \in \mathcal{M}} \bar{a}_{n,m}^2 \bar{\eta}_{n,m} \mathbf\zeta_{n,m}  \leq 1.
\end{equation}
Accordingly, the received signal at the $u$-th unicast user and the $k_m$-th multicast user  are written as
\begin{equation}
\begin{split}\label{eq:ru}
r_{u}=&\sqrt{p_{\text{dl}}} \sum_{n \in \mathcal{N}} a_{n,u} \sqrt{\eta_{n,u}} \mathbf{c}_{n,u}^T \mathbf{\hat{c}^*}_{n,u} q_u\\
&+ \sqrt{p_{\text{dl}}} \sum_{\Acute{u} \in \mathcal {U}, \Acute{u} \neq u} \sum_{n \in \mathcal{N}} a_{n,\Acute{u}} \sqrt{\eta_{n,\Acute{u}}} \mathbf{c}_{n,u}^T \mathbf{\hat{c}^*}_{n,\Acute{u}} q_{\Acute{u}} \\ 
&+ \sqrt{p_{\text{dl}}} \sum_{m \in \mathcal{M}} \sum_{n \in \mathcal{N}} \bar{a}_{n,m} \sqrt{\bar{\eta}_{n,m}} \mathbf{c}_{n,u}^T \mathbf{\hat{t}^*}_{n,m} \bar{q}_m +{w}_{u}, 
\end{split}
\end{equation}
\begin{equation}\label{eq:rmk}
\begin{split}
r_{m,k}&=\sqrt{p_{\text{dl}}} \sum\nolimits_{n \in \mathcal{N}} \bar{a}_{n,m} \sqrt{\bar{\eta}_{n,m}} \mathbf{t}_{n,m,k}^T \mathbf{\hat{t}^*}_{n,m} \bar{q}_m \\
&+ \sqrt{p_{\text{dl}}} \sum\nolimits_{\Acute{m} \in \mathcal{M}, \Acute{m} \neq m} \sum\nolimits_{n \in \mathcal{N}} \bar{a}_{n,\Acute{m}} \sqrt{\bar{\eta}_{n,\Acute{m}}} \mathbf{t}_{n,m,k}^T \mathbf{\hat{t}^*}_{n,\Acute{m}} \bar{q}_{\Acute{m}}  \\ & + \sqrt{p_{\text{dl}}} \sum\nolimits_{u \in U} \sum\nolimits_{n \in \mathcal{N}} a_{n,u} \sqrt{\eta_{n,{u}}} \mathbf{t}_{n,m,k}^T \mathbf{\hat{c}^*}_{n,{u}} q_{{u}} + w_{m,k},
\end{split}
\end{equation}
respectively.
Using the popular use-and-then-forget bounding technique\cite{Hien:JWCOM:2017}, the SE at the $u$-th unicast user and  $k_m$-th multicast user can be derived as stated in the following proposition.

\begin{proposition}~\label{Theorem3}
The   SE expressions for the $u$-th unicast user and $k_m$-th multicast users are given by
$\mathrm{SE}_{u} = \frac{T-\tau}{T} \log_{2}\Big({1+\mathrm{SINR}_u} \Big)$
and
$\mathrm{SE}_{m,k}  = \frac{T-\tau}{T} \log_{2}\Big({1+ \mathrm{SINR}_{m,k}} \Big),$
respectively, where the closed-form expressions for the received SINR at the $u$-th unicast user, $\mathrm{SINR}_u$,  and at the $k_m$-th multicast user, $\mathrm{SINR}_{m,k}$,  are given by \eqref{eq:sinru} and \eqref{eq:sinrmk} at the top of the  page, respectively.
\end{proposition}
 \begin{proof}
     The proof is omitted due to page constraints. 
 \end{proof}
 Note that, unlike \cite{Li:TVT:2022}, where the SE expressions are derived in approximate forms, our expressions are in exact form.
\vspace{-0.5em}
\section{Problem Formulation and APG Optimization}
Here, we formulate an optimization problem to determine the power allocation coefficients   $\boldsymbol{\eta} \triangleq \{\eta_{n,u},\bar{\eta}_{n,m}\} $ and   user association $\boldsymbol{a} \triangleq \{a_{n,u},\bar{a}_{n,m}\}$ in order to maximize the weighted SSE for both unicast and multicast users, subject to  SE requirements and the per-AP transmit power constraint~\eqref{eq:Exn}  as
\setcounter{equation}{11}
\begin{subequations} \label{eq:optp}
\begin{align}
&\max\limits_{\boldsymbol{a,\eta} } \Big(w_1 \sum_{u \in \mathcal {U}} \mathrm{SE}_{u} + w_2 \sum_{m \in \mathcal{M}} \sum_{k \in\mathcal{K}_m} \mathrm{SE}_{m,k} \Big)\\
\text {s.t.} ~~
&C_1: \mathrm{SE}_{u} \ge \mathrm{SE}_{QoS} ~,~ \mathrm{SE}_{m,k} \ge \mathrm{\bar{SE}}_{QoS}, ~ \forall ~u, m, k, \label{eq:consta}\\
&C_2: \eta_{n,u} \geq 0, ~,~  \bar{\eta}_{n,m} \geq 0, ~~ \forall ~n,u,m, \label{eq:constb} \\
&C_3: L \sum_{u \in \mathcal {U}} a_{n,u}^2 \eta_{n,u} \mathbf\gamma_{n,u}  + L \sum_{m \in \mathcal{M}} \bar{a}_{n,m}^2 \bar{\eta}_{n,m} \mathbf\zeta_{n,m}  \leq 1, ~\forall n, \label{eq:constc}\\
& C_4: \sum\nolimits_{n \in \mathcal{N}} a_{n,u} \ge 1,~ \sum\nolimits_{n \in \mathcal{N}} \bar{a}_{n,m} \ge 1,~~ \forall ~ u,m, \label{eq:constd} \\
& C_5: \sum\nolimits_{u \in \mathcal{U}} a_{n,u} + \sum\nolimits_{m \in \mathcal{M}} \bar{a}_{n,m} \leq K_{\mathrm{max}},~~ \forall ~ m\label{eq:constk}, 
\end{align}
\end{subequations}
where $w_1$ and $w_2$,  $w_1+w_2=1$, are the weighting coefficients, while $\mathrm{SE}_{QoS}$ and $\mathrm{\bar{SE}}_{QoS}$ in~\eqref{eq:consta} denote the minimum SE requirements for the  unicast user $u$ and  multicast user $k_m$, respectively, to ensure QoS in the network. The constraint \eqref{eq:constd} guarantees that at least one AP serves each unicast and multicast user, while constraint~\eqref{eq:constk} guarantees that   the maximum number of unicast user and multicast group served by each AP is $ K_{\mathrm{max}}, 1 \leq K_{\mathrm{max}} \leq U+M$.
\begin{figure*}
\begin{equation} \label{eq:sinruth}
\mathrm{SINR}_{u}(\boldsymbol{\theta})\triangleq\frac{U_u(\boldsymbol{\theta})}{V_u(\boldsymbol{\theta})}=\frac{(  \sqrt{p_{\text{dl}}}~L~\sum_{n \in \mathcal{N}}  \theta_{n,u} \sqrt{\mathbf\gamma_{n,u}})^2}{ p_{\text{dl}}~ L\sum_{u' \in \mathcal {U}} \sum_{n \in \mathcal{N}}  \theta^2_{n,u'} \beta_{n,u} + p_{\text{dl}} ~ L \sum_{m \in \mathcal{M}} \sum_{n \in \mathcal{N}} \bar{\theta}^2_{n,m}   \beta_{n,u} +1}~\tag{14}
\end{equation}
\begin{equation} \label{eq:sinrmkth}
\mathrm{SINR}_{m,k}(\boldsymbol{\theta})\triangleq\frac{U_{m,k}(\boldsymbol{\theta})}{V_{m,k}(\boldsymbol{\theta})}=\frac{( \sqrt{p_{\text{dl}}} ~ L\sum_{n \in \mathcal{N}} \bar{\theta}_{n,m} ~ \frac{\mathbf\xi_{n,m,k}}{\sqrt{\zeta_{n,m}}} )^2}{p_{\text{dl}}  ~L~\sum_{m' \in \mathcal{M}} \sum_{n \in \mathcal{N}} \bar{\theta}^2_{n,m'}\lambda_{n,m,k} + p_{\text{dl}} ~L \sum_{u \in \mathcal {U}} \sum_{n \in \mathcal{N}}  \theta^2_{n,u} \lambda_{n,m,k}  +1}~\tag{15}
\end{equation}
\hrulefill
	\vspace{-4mm}
\end{figure*}
\subsection{APG Method and Problem Reformulation}
The joint optimization problem~\eqref{eq:optp} is non-convex mixed-integer problem that is difficult to solve.
Here, we leverage the APG approach to tackle our joint optimization problem  \eqref{eq:optp}. Although the APG approach is suboptimal, it offers significantly lower complexity compared to common SCA algorithms, especially beneficial for handling large-scale CF-mMIMO networks~\cite{Farooq:TCOM:2021}\cite{Mai:TWC:2022}. 
To facilitate algorithmic design, we first define new variables and then convert problem~\eqref{eq:optp} into a more tractable form, as outlined below:

\begin{itemize}
    \item  Let $\boldsymbol{\theta}\triangleq[\boldsymbol{\theta}^T_1,\ldots,\boldsymbol{\theta}^T_N]^T$, where $\boldsymbol{\theta}_n=[\theta_{n,1},\dots,\theta_{n,U},\bar{\theta}_{n,1},\dots,\bar{\theta}_{n,M}]^T$, while $\theta_{n,u}=\sqrt{\eta_{n,u}\gamma_{n,u}}$ and $\bar{\theta}_{n,m}=\sqrt{\bar{\eta}_{n,m}\zeta_{n,m}}$.
\item By considering~\eqref{eq:anu} and~\eqref{eq:anm},  we enforce
\begin{align} \label{eq:thataa}
   \theta_{n,u} &= 0 , \mathrm{if}~a_{n,u}= 0 ~\forall  ~ n,u,  \nonumber\\ \bar{\theta}_{n,m} &= 0 , \mathrm{if}~\bar{a}_{n,m}= 0 ~\forall  ~ n,m,
\end{align}
to ensure that if AP $n$ does not associate with unicast  user $u$ (multicast user $k_m$), the transmit power ${  p_{\text{dl}} \theta_{n,u}^2}/{\gamma_{n,u}}$  towards  unicast  user $u$ (${  p_{\text{dl}} \bar{\theta}_{n,m}^2}/{\zeta_{n,m}}$   towards  multicast user $k_m$) is zero.
\item Now, we rewrite $\mathrm{SINR}_u$ and   $\mathrm{SINR}_{m,k}$   as a function of $\boldsymbol{\theta}$, as  \eqref{eq:sinruth} and \eqref{eq:sinrmkth}, at the top of the next  page, respectively. We highlight that the user association  $\boldsymbol{a}$ only affects the SE expressions  via parameter $\boldsymbol{\theta}$  and~\eqref{eq:thataa}.
\end{itemize}
Thus, optimization problem~\eqref{eq:optp} can be reformulated as
\setcounter{equation}{15}
\begin{subequations}\label{eq:optpth}
\begin{align}
&\min\limits_{\boldsymbol{a,\theta} } -\Big(w_1 \sum_{u \in \mathcal {U}}\mathrm{SE}_{u} (\boldsymbol{\theta}) + w_2 \sum_{m \in \mathcal{M}} \sum_{k \in \mathcal{K}_m} \mathrm{SE}_{m,k} (\boldsymbol{\theta}) \Big) \\
\text {s.t.} ~
&C_1: \mathrm{SE}_{u} (\boldsymbol{\theta}) \ge \mathrm{SE}_{QoS} , \mathrm{SE}_{m,k} (\boldsymbol{\theta}) \ge \mathrm{\bar{SE}}_{QoS}, \forall u, m, \label{eq:consttha}\\
&C_2: \theta_{n,u}  \geq 0 ~,~  \bar{\theta}_{n,m} \geq 0, ~~ \forall ~n,u,m, \label{eq:constthb} \\
&C_3: L \sum\nolimits_{u \in \mathcal {U}} \theta_{n,u}^2   + L \sum\nolimits_{m \in \mathcal{M}} \bar{\theta}_{n,m}^2   \leq 1,  ~ \forall ~n, \label{eq:constthc}\\
& C_4: \sum\nolimits_{n \in \mathcal{N}} a_{n,u} \ge 1,~ \sum\nolimits_{n \in \mathcal{N}} \bar{a}_{n,m} \ge 1,~~ \forall ~ u,m, \label{eq:constthd} \\
& C_5: \sum\nolimits_{u \in \mathcal{U}} a_{n,u} + \sum\nolimits_{m \in \mathcal{M}} \bar{a}_{n,m} \leq K_{\mathrm{max}},~~ \forall ~ m. \label{eq:constthe}
\end{align}
\end{subequations}
To address the the binary constraint in $\eqref{eq:anu}$ and $\eqref{eq:anm}$,  we note that $x \in$ $\{0,1\} \Leftrightarrow x \in[0,1]\quad \& \quad x-x^2 \leq 0$\cite{Vu:IOT:2022},  and hence, we  replace $\eqref{eq:anu}$ and $\eqref{eq:anm}$  with the following AP association constraints
\begin{equation}\label{eq:conua}
S_u(\mathbf{a}) \triangleq \sum_{u \in \mathcal {U}} \!  \sum_{n \in \mathcal{N}} (a_{n,u}-a_{n,u}^2) \leq 0, ~~ 0\! \leq a_{n,u}\! \leq 1, ~\forall n, u,
\end{equation}
\begin{equation}\label{eq:conma}
\bar{S}_{m}(\mathbf{a})   \!\!\triangleq \!\! \sum_{m \in \mathcal{M}}  \! \sum_{n \in \mathcal{N}} (\bar{a}_{n,m}-\bar{a}_{n,m}^2) \leq 0,~~
0 \!\leq \bar{a}_{n,m} \!\leq 1,~ \forall n, m,
\end{equation}
respectively.
Thus,
\begin{equation}\label{eq:tha}
\begin{split}
\theta_{n,u}^2 \leq a_{n,u}, ~~~~~
\bar{\theta}_{n,m}^2 \leq \bar{a}_{n,m}.
\end{split}
\end{equation}
Now, we define the new parameter $\boldsymbol{z}\triangleq[\boldsymbol{z}^T_1,\ldots,\boldsymbol{z}^T_N]^T$,   where
$\boldsymbol{z}_n=[z_{n,1},\dots,z_{n,U},\bar{z}_{n,1},\dots,\bar{z}_{n,M}]^T$,
while  $z_{n,u}^2 \triangleq a_{n,u}$ and $\bar{z}_{n,m}^2 \triangleq \bar{a}_{n,m}$ with 
\begin{equation}\label{eq:znuznm}
\begin{split}
0\leq z_{n,u} \leq 1 ~~\text{and} ~~0\leq \bar{z}_{n,m}\leq 1.
\end{split}
\end{equation}
Therefore, constraint $C_5$ in~\eqref{eq:optp} can be re-expressed as
\vspace{0.2 em}
\begin{align}\label{eq:znuznmkh}
    \sum\nolimits_{u \in \mathcal{U}} z_{n,u}^2 + \sum\nolimits_{m \in \mathcal{M}} \bar{z}_{n,m}^2 \leq K_{\mathrm{max}},~~ \forall ~ m.
\end{align}
In addition, constraints \eqref{eq:consttha}, \eqref{eq:conua}, \eqref{eq:conma},  \eqref{eq:constthd}, and \eqref{eq:tha} can be replaced by
\vspace{0.2 em}
\begin{equation} \label{eq:c1uc1m}
\begin{split}
&C_{1,u}(\boldsymbol{\theta}) \triangleq \sum_{u \in \mathcal {U}} [\max (0,\mathrm{SE}_{QoS} -\mathrm{SE}_{u}(\boldsymbol{\theta}))]^2\leq 0,\\
&\bar{C}_{1,m}(\boldsymbol{\theta}) \triangleq \sum_{m \in \mathcal{M}} \sum_{k \in \mathcal{K}_m} [\max (0,\mathrm{\bar{SE}}_{QoS}-\mathrm{SE}_{m,k}(\boldsymbol{\theta}))]^2 \leq 0,
\end{split}
\end{equation}
\begin{align}\label{eq:c2uc2m}
&C_{2,u}(\mathbf{z}) \triangleq \sum_{u \in \mathcal {U}} \sum_{n \in \mathcal{N}}(z_{n,u}^2-z_{n,u}^4) \leq 0, \nonumber\\
&\bar{C}_{2,m}(\mathbf{z}) \triangleq \sum_{m \in \mathcal{M}} \sum_{n \in \mathcal{N}} (\bar{z}_{n,m}^2-\bar{z}_{n,m}^4) \leq 0,
\end{align}

\begin{align} 
C_{3,u}(\boldsymbol{\theta}, \mathbf{z}) \triangleq &\sum_{u \in \mathcal {U}} ([\max(0,1-\sum_{n \in \mathcal{N}} z_{n,u}^2)]^2\nonumber\\
&+\sum_{n \in \mathcal{N}} [\max (0, \theta_{n,u}^2 - z_{n,u}^2)]^2)\leq 0,
\end{align}
\begin{align}
\bar{C}_{3,m}(\boldsymbol{\theta}, \mathbf{z}) \triangleq& \sum_{m \in \mathcal{M}} ([\max(0,1-\sum_{n \in \mathcal{N}} \bar{z}_{n,m}^2)]^2\nonumber\\
&+\sum_{n \in \mathcal{N}} [\max (0, \bar{\theta}_{n,m}^2 - \bar{z}_{n,m}^2)]^2)\leq 0.
\end{align}
Now, we define
\begin{align}\label{eq:fv}
&g(\mathbf{\vartheta}) \triangleq  -\Big(w_1 \sum_{u \in \mathcal {U}} \mathrm{SE}_{u}(\boldsymbol{\theta}) + w_2 \sum_{m \in \mathcal{M}} \sum_{k \in \mathcal{K}_m}\mathrm{SE}_{m,k}(\boldsymbol{\theta}) \Big) \nonumber \\ 
&+ X \Big[\mu_1 (C_{1,u}(\boldsymbol{\theta})+\bar{C}_{1,m}(\boldsymbol{\theta}))+\mu_2(C_{2,u}(\mathbf{z})+\bar{C}_{2,m}(\mathbf{z})) \nonumber \\
&+\mu_3(C_{3,u}(\boldsymbol{\theta}, \mathbf{z})+\bar{C}_{3,m}(\boldsymbol{\theta}, \mathbf{z}))\Big],
\end{align}
where $\mu_1,\mu_2 ~\text{and}~ \mu_3$ are positive weights,  $X$ is the Lagrangian multiplier, and $\mathbf{\vartheta} \triangleq\left[\boldsymbol{\theta}^T, \mathbf{z}^T\right]^T$. 
Thus, the optimization problem~\eqref{eq:optpth} can be expressed equivalently as
\begin{equation}\label{eq:apgopt}
\min\limits_{\mathbf{\vartheta} \in \widehat{\mathcal{C}}} g(\mathbf{\vartheta}),
\end{equation}
where $\widehat{\mathcal{C}} \triangleq \{\eqref{eq:constthb}, \eqref{eq:constthc}, \eqref{eq:znuznm},\eqref{eq:znuznmkh}\}$ is the convex feasible set. Our proposed method to solve problem \eqref{eq:apgopt} is given in \textbf{Algorithm \ref{algo}. }
The primary tasks in executing \textbf{Algorithm 1} include computing the gradient of the objective function and performing projections, as outlined below:
\vspace{-0.4em}
\subsubsection{Gradient  of $g(\mathbf{\vartheta})$}
The gradients  $\frac{\partial}{\partial \theta_{n,u}} g(\mathbf{\vartheta})$ and $\frac{\partial}{\partial z_{n,u}} g(\mathbf{\vartheta})$ is given by
\begin{equation}
 \frac{\partial}{\partial \theta_{n,u}} g(\mathbf{\vartheta}) = -w_1\sum_{i \in \mathcal {U}} \frac{\partial}{\partial \theta_{n,u}} \mathrm{SE}_{i}(\boldsymbol{\theta}) + X \frac{\partial}{\partial \theta_{n,u}} C_{u} (\mathbf{\vartheta}),
 \end{equation}
 \begin{equation}
\frac{\partial}{\partial z_{n,u}} g(\mathbf{\vartheta})= -w_1\sum_{i \in \mathcal {U}} \frac{\partial}{\partial z_{n,u}}\mathrm{SE}_{i}(\boldsymbol{\theta}) + X \frac{\partial}{\partial z_{n,u}} C_{u} (\mathbf{\vartheta}),
\end{equation}
with $C_u(\mathbf{\vartheta})= \mu_1 C_{1,u}(\boldsymbol{\theta})+\mu_2 C_{2,u}(\mathbf{z})+\mu_3 C_{3,u}(\boldsymbol{\theta}, \mathbf{z})$,

while $\frac{\partial}{\partial \theta_{n,u}} \mathrm{SE}_{i}(\boldsymbol{\theta})$ is given by
\begin{equation}
\frac{\partial}{\partial \theta_{n,u}} \mathrm{SE}_{i}(\boldsymbol{\theta})
=\frac{T-\tau}{T\ln 2}\bigg[\frac{\frac{\partial}{\partial \theta_{n,u}}(U_i(\boldsymbol{\theta})+V_i(\boldsymbol{\theta}))}{(U_i(\boldsymbol{\theta})+V_i(\boldsymbol{\theta}))}-\frac{\frac{\partial}{\partial \theta_{n,u}} V_i(\boldsymbol{\theta})}{V_i(\boldsymbol{\theta})}\bigg],
\end{equation}
with
\begin{equation}
\begin{split}
&\frac{\partial U_i(\boldsymbol{\theta})}{\partial \theta_{n,u}} = 
\begin{cases}
2( \sqrt{p_{\text{dl}}}L\sum_{n \in \mathcal{N}} \theta_{n,u} \sqrt{\mathbf\gamma_{n,u}})(\sqrt{p_{\text{dl}}}L \sqrt{\mathbf\gamma_{n,u}}), & i=u, \\
0, 
& i \neq u, 
\end{cases} \\
&\frac{\partial}{\partial \theta_{n,u}} V_i(\boldsymbol{\theta})= 
\begin{cases}
2~ p_{\text{dl}} L \theta_{n,u} \beta_{n,u}, & i = u , \\
2~ p_{\text{dl}} L \theta_{n,u} \beta_{n,i}, & i \neq u.
\end{cases}
\end{split}  
\end{equation}
\vspace{-0.5em}
Also,
$ -\sum_{i \in \mathcal {U}} \frac{\partial}{\partial z_{n,u}} \mathrm{SE}_{i}(\boldsymbol{\theta}) = 0, ~ \forall~ n, u,i$. In addition, 
\begin{align}
\frac{\partial}{\partial \theta_{n,u}} C_{u} (\mathbf{\vartheta}) =& - \mu_1 \sum_{i \in \mathcal {U}} 2  \max (0,\mathrm{SE}_{QoS}-\mathrm{SE}_{i}(\boldsymbol{\theta})) \nonumber \\
&\hspace{-4em}\times \frac{\partial}{\partial \theta_{n,u}} \mathrm{SE}_{i}(\boldsymbol{\theta})+ 4 \mu_3  \max (0, \theta_{n,u}^2 - z_{n,u}^2)\theta_{n,u}, 
\end{align}
and
\begin{align}
& \frac{\partial}{\partial z_{n,u}} C_{u} (\mathbf{\vartheta}) =\mu_2 (2z_{n,u}-4z_{n,u}^3) - 4 \mu_3 (0, \theta_{n,u}^2\nonumber\\
&\hspace{1em}- z_{n,u}^2)z_{n,u}- 4 \mu_3 \max(0,1-\sum_{n \in \mathcal{N}} z_{n,u}^2)z_{n,u},
\end{align}
while $\frac{\partial}{\partial \bar{\theta}_{n,m}} g(\mathbf{\vartheta})$ and $\frac{\partial}{\partial \bar{z}_{n,m}} g(\mathbf{\vartheta})$ is given by
\begin{equation}
 \frac{\partial}{\partial \bar{\theta}_{n,m}} g(\mathbf{\vartheta})=  -w_2\sum_{i \in \mathcal{M}} \sum_{k \in \mathcal{K}_m} \frac{\partial}{\partial \bar{\theta}_{n,m}} \mathrm{SE}_{i,k}(\boldsymbol{\theta}) + X \frac{\partial}{\partial \bar{\theta}_{n,m}} \bar{C}_{m} (\mathbf{\vartheta}),
 \end{equation}
 \begin{equation}
\frac{\partial}{\partial \bar{z}_{n,m}} g(\mathbf{\vartheta})= -w_2\sum_{i \in \mathcal{M}}\sum_{k \in \mathcal{K}_m} \frac{\partial}{\partial \bar{z}_{n,m}} \mathrm{SE}_{i,k} (\boldsymbol{\theta}) + X \frac{\partial}{\partial \bar{z}_{n,m}} \bar{C}_{m} (\mathbf{\vartheta}),
\end{equation}
where $\bar{C}_m(\mathbf{\vartheta})= \mu_1 \bar{C}_{1,m}(\boldsymbol{\theta})+\mu_2 \bar{C}_{2,m}(\mathbf{z})+\mu_3 \bar{C}_{3,m}(\boldsymbol{\theta}, \mathbf{z})$. 
On the other hand $\frac{\partial}{\partial \bar{\theta}_{n,m}}  \mathrm{SE}_{i,k}(\boldsymbol{\theta})$ is calculated as
\begin{equation}
\begin{split}
\frac{\partial}{\partial \bar{\theta}_{n,m}} & \mathrm{SE}_{i,k}(\boldsymbol{\theta})\\
&=\frac{T-\tau}{T\ln 2}\bigg[\frac{\frac{\partial}{\partial \bar{\theta}_{n,m}}(U_{i,k}(\boldsymbol{\theta})+V_{i,k}(\boldsymbol{\theta}))}{(U_{i,k}(\boldsymbol{\theta})+V_{i,k}(\boldsymbol{\theta}))}-\frac{\frac{\partial}{\partial \bar{\theta}_{n,m}} V_{i,k}(\boldsymbol{\theta})}{V_{i,k}(\boldsymbol{\theta})}\bigg],
\end{split}
\end{equation}
with
\begin{equation}
\begin{split}
&\frac{\partial}{\partial \bar{\theta}_{n,m}} U_{i,k}(\boldsymbol{\theta})= 
\begin{cases}
2L^2p_{\text{dl}}(\sum_{n \in \mathcal{N}} \bar{\theta}_{n,m} \frac{\mathbf\xi_{n,m,k}}{\sqrt{\zeta_{n,m}}})(\frac{\mathbf\xi_{n,m,k}}{\sqrt{\zeta_{n,m}}}) & i=m, \\
0, & i \neq m, 
\end{cases} \\
&\frac{\partial}{\partial \bar{\theta}_{n,m}} V_{i,k}(\boldsymbol{\theta})= 
\begin{cases}
2~ p_{\text{dl}}~ L ~ \bar{\theta}_{n,m}\lambda_{n,m,k}, & i = m , \\
2~ p_{\text{dl}}~ L ~ \bar{\theta}_{n,m}\lambda_{n,i,k}, & i \neq m,
\end{cases}
\end{split}  
\end{equation}
while
$ -\sum_{i \in \mathcal{M}} \sum_{k \in \mathcal{K}_m} \frac{\partial}{\partial \bar{z}_{n,m}} \mathrm{SE}_{i,k}(\boldsymbol{\theta}) = 0, ~ \forall~ n,i,k$. In addition,
\begin{equation}
\begin{split}
&\frac{\partial}{\partial \bar{\theta}_{n,m}} \bar{C}_{m} (\mathbf{\vartheta})
= - \mu_1 \sum_{m \in \mathcal{M}} \sum_{k \in \mathcal{K}_m} 2  \max (0,\mathrm{\bar{SE}}_{QoS}-\mathrm{SE}_{i}(\boldsymbol{\theta})) \nonumber\\
&\times\frac{\partial}{\partial \bar{\theta}_{n,m}} \mathrm{SE}_{i}(\boldsymbol{\theta})+ 4 \mu_3 \max(0, \bar{\theta}_{n,m}^2 - \bar{z}_{n,m}^2)\bar{\theta}_{n,m},
\end{split}  
\end{equation}
\begin{align}
\frac{\partial}{\partial \bar{z}_{n,m}} \bar{C}_{m} (\mathbf{\vartheta}) =& \mu_2 (2\bar{z}_{n,m}-4\bar{z}_{n,m}^3) - 4 \mu_3 (0, \bar{\theta}_{n,m}^2 - \bar{z}_{n,m}^2)\bar{z}_{n,m}
\nonumber\\
&- 4 \mu_3 \max(0,1-\sum_{n \in \mathcal{N}} \bar{z}_{n,m}^2)\bar{z}_{n,m}.
\end{align}
\vspace{-0.5em}
\subsubsection{Projection onto $\widehat{\mathcal{C}}$}
The projection of the   given  $\mathbf{r}  \in \mathbb{R}^{2N(U+M) \times 1}$  onto the feasible set $\widehat{\mathcal{C}}$ in  \textbf{Step 5} of   \textbf{Algorithm~\ref{algo}}
 can be done by solving the  problem
\begin{align} \label{eq:proj}
&\mathcal{P}_{\widehat{\mathcal{C}}}(\mathbf{r}):\min _{\mathbf{\vartheta} \in \mathbb{R}^{2N(U+M) \times 1}}\|\mathbf{\vartheta}-\mathbf{r}\|^2 \\
&\text { s.t. } ~~ \eqref{eq:constthb},\eqref{eq:constthc},\eqref{eq:znuznm}, \eqref{eq:znuznmkh}
\end{align}
with $\mathbf{r}=\left[\mathbf{r}_1^T, \mathbf{r}_2^T\right]^T$, $\mathbf{r}_1 \triangleq\left[\mathbf{r}_{1,1}^T, \ldots, \mathbf{r}_{1, N}^T\right]^T$ and $\mathbf{r}_{1, n} \triangleq\left[r_{1, n1}, \ldots, r_{1, nU} ,\bar{r}_{1, n1},\ldots,\bar{r}_{1, nM} \right]^T$, while $\mathbf{r}_{2, n} \triangleq\left[r_{2, n1} , \ldots, r_{2, nU} ,\bar{r}_{2, n1},\ldots,\bar{r}_{2, nM} \right]^T$. Problem \eqref{eq:proj} can be split into two separate sub problems for calculating $\boldsymbol{\theta}_n $ and $\mathbf{z}_n$. Following a similar approach as in \cite{Farooq:TCOM:2021},  we can find the following closed-form expressions 
\begin{align} \label{eq:65}
\boldsymbol{\theta}_n=\frac{1/\sqrt{L}}{\max \left(1/\sqrt{L},\left\|\left[\mathbf{r}_{1, n}\right]_{0}^+\right\|\right)}\left[\mathbf{r}_{1, n}\right]_{0}^+,
\end{align}
where $ [\Psi]_{0}^+\! \triangleq[\max \left(0, \psi_1 \right), \ldots, \max \left(0, \psi_{U} \right),
\max \left(0, \bar{\psi}_1\right),\dots,\\\max \left(0, \bar{\psi}_{M}\right) ]^T, \forall \Psi \in \mathbb{R}^{(U+M) \times 1}$
and
\vspace{0.2 em}
 \begin{align}
\mathbf{z}_n=\bigg[\frac{\sqrt{K_{\mathrm{max}}}}{\max \left(\sqrt{K_{\mathrm{max}}},\left\|\left[\mathbf{r}_{2, n}\right]_{0}^+\right\|\right)}\left[\mathbf{r}_{2, n}\right]_{0}^+\bigg]_{1-},
\end{align}
where 
$ [\Psi]_{1-} \triangleq [\min \left(1, \psi_1 \right), \ldots, \min \left(1, \psi_{U} \right),
\min \left(1, \bar{\psi}_1\right),\dots,\\\min \left(1, \bar{\psi}_{M}\right) ]^T, \forall \Psi \in \mathbb{R}^{(U+M) \times 1}.$
Given that the feasible set $\widehat{\mathcal{C}}$ is bounded, it follows that $\nabla g(\mathbf{\vartheta})$ is Lipschitz continuous with a known constant $J$. This implies that for all $\mathbf{v}, \mathbf{w} \in \widehat{\mathcal{C}}$, the gradient satisfies $\|\nabla g(\mathbf{v})-\nabla g(\mathbf{w})\| \leq J\|\mathbf{v}-\mathbf{w}\|$.

In \textbf{Algorithm~\ref{algo}}, beginning with a random initial point $\overline{\mathbf{\vartheta}}^{(0 )}$, we update $\overline{\mathbf{\vartheta}}^{(o )}$ at each iteration as follows:
\begin{align} \label{eq:Vo}
\overline{\mathbf{\vartheta}}^{(o)} = \mathbf{\vartheta}^{(o)} + \frac{q^{(o - 1)}}{q^{(o)}}(\tilde{\mathbf{\vartheta}}^{(o)} - \mathbf{\vartheta}^{(o)}) + \frac{q^{(o - 1)} - 1}{q^{(o)}}(\mathbf{\vartheta}^{(o)} - \mathbf{\vartheta}^{(o - 1)}),
\end{align}
where 
\begin{align} \label{eq:q_n}
q^{(o+1)}=\frac{1+\sqrt{4\left(q^{(o)}\right)^2+1}}{2}.
 \end{align}
We then proceed along the gradient of the objective function with a specified step size $\alpha_{\overline{\mathbf{\vartheta}}}$. The resulting point $\left(\overline{\mathbf{\vartheta}}-\alpha_{\overline{\mathbf{\vartheta}}} \nabla g(\overline{\mathbf{\vartheta}})\right)$ is subsequently projected onto the feasible set $\widehat{\mathcal{C}}$, yielding
\begin{align} 
\tilde{\mathbf{\vartheta}}^{(o+1)}=\mathcal{P}_{\hat{\mathcal{C}}}\left(\overline{\mathbf{\vartheta}}^{(o)}-\alpha_{\overline{\mathbf{\vartheta}}} \nabla g\left(\overline{\mathbf{\vartheta}}^{(o)}\right)\right).
\end{align}
It is important to note that $g(\mathbf{\vartheta})$ is not convex, so $g\left(\tilde{\mathbf{\vartheta}}^{(o+1)}\right)$ may not necessarily improve the objective sequence. Consequently, $\mathbf{\vartheta}^{(o+1)} = \tilde{\mathbf{\vartheta}}^{(o+1)}$ is accepted only if the objective value $g\left(\tilde{\mathbf{\vartheta}}^{(o+1)}\right)$ is below $c^{(o)}$, which acts as a relaxation of $g\left(\mathbf{\vartheta}^{(o)}\right)$ while remaining relatively close to it.
$c^{(o)}$ can be computed as 
\vspace{0.3em}
\begin{align} 
& c^{(o)}= \frac{\sum\nolimits_{o=1}^\kappa \zeta^{(\kappa-o)} g\left(\mathbf{\vartheta}^{(o)}\right)}{\sum\nolimits_{o=1}^\kappa \zeta^{(\kappa-o)}},
\end{align}
where $\zeta \in[0,1)$ used to control the non-monotonicity degree. 
After each iteration, $c^{(o)}$ can be updated iteratively as follows
\vspace{0.3em}
\begin{align}\label{eq:cn+1}
&{c}^{(o+1)}=\frac{\zeta b^{(o)} c^{(o)}+g\left(\mathbf{\vartheta}^{(o)}\right)}{b^{(o+1)}},
\end{align}
where $c^{(1)}\!=\!g\left(\mathbf{\vartheta}^{(1)}\right)$,  $b^{(1)}\!=\!1$, and ${b}^{(o+1)}$ can be  obtained as
\vspace{0.3em}
\begin{align}\label{eq:bn+1}
&{b}^{(o+1)}=\zeta b^{(o)}+1. 
\end{align}
When the condition $g\left(\tilde{\mathbf{\vartheta}}^{(o+1)}\right) \leq c^{(o)}-\zeta\Vert\tilde{\mathbf{\vartheta}}^{(o+1)}-\overline{\mathbf{\vartheta}}^{(o)}\Vert^2$ is not satisfied, extra correction steps are employed to avoid this situation. Specifically, another point is calculated with a dedicated step size $\alpha_{\mathbf{\vartheta}}$ as 
\vspace{0.3em}
\begin{align} \label{eq:vnhat+1}
\hat{\mathbf{\vartheta}}^{(o+1)}=\mathcal{P}_{\hat{\mathcal{C}}}\left(\mathbf{\vartheta}^{(o)}-\alpha_{\mathbf{\vartheta}} \nabla g\left(\mathbf{\vartheta}^{(o)}\right)\right).
\end{align}
Then, $\mathbf{\vartheta}^{(o+1)}$ is updated by comparing the objective values at $\tilde{\mathbf{\vartheta}}^{(o+1)}$ and $\hat{\mathbf{\vartheta}}^{(o+1)}$ as
\vspace{0.4em}
\begin{align}\label{eq:vn+1}
\mathbf{\vartheta}^{(o+1)} \triangleq\left\{\begin{array}{ll}
\tilde{\mathbf{\vartheta}}^{(o+1)}, & \text { if } g\left(\tilde{\mathbf{\vartheta}}^{(o+1)}\right) \leq g\left(\hat{\mathbf{\vartheta}}^{(o+1)}\right),  \\
\hat{\mathbf{\vartheta}}^{(o+1)}, & \text {otherwise}.
\end{array}\right.
\end{align}
Finally, we emphasize that our proposed APG-based optimization approach operates on the large-scale fading timescale, which varies slowly over time.

\begin{algorithm}[!t]
\caption{Solving \eqref{eq:apgopt} Using APG Approach}
\begin{algorithmic}[1]\label{algo}
\STATE \textbf{Initialize}: $X$, $ \varsigma>1$, $ \vartheta >0$, $\mathbf{\vartheta}^{(0)}$,  $\alpha_{\overline{\mathbf{\vartheta}}}>0, \alpha_{\mathbf{\vartheta}}>0$. Set $\tilde{\mathbf{\vartheta}}^{(1)}=\mathbf{\vartheta}^{(1)}=\mathbf{\vartheta}^{(0)}, \zeta \in[0,1), b^{(1)}=1, c^{(1)} = g(\mathbf{\vartheta}^{(1)})$,   $n=1, q^{(0)}=0, q^{(1)}=1$. Choose $\overline{\mathbf{\vartheta}}^{(0)}$ from feasible set $ \widehat{\mathcal{C}}$.
\STATE \textbf{repeat}  
\STATE \textbf{while}   $\big|\frac{g\left(\mathbf{\vartheta}^{(o)}\right)-g\left(\mathbf{\vartheta}^{(o-10)}\right)}{g\left(\mathbf{\vartheta}^{(o)}\right)}\big|\!\leq\! \epsilon$ \text{or} $\big|\frac{f\left(\boldsymbol{\theta}^{(o)}\right)-f\left(\boldsymbol{\theta}^{(o-1)}\right)}{f\left(\boldsymbol{\theta}^{(o)}\right)}\big|\!\leq\! \epsilon$ \textbf{do}
\STATE \text { update } $\overline{\mathbf{\vartheta}}^{(o)}$ as~\eqref{eq:Vo}
\STATE Set $\tilde{\mathbf{\vartheta}}^{(o+1)}=\mathcal{P}_{\hat{\mathcal{C}}} (\overline{\mathbf{\vartheta}}^{(o)}-\alpha_{\mathbf{\vartheta}} \nabla g (\overline{\mathbf{\vartheta}}^{(o)}))$,
\STATE \textbf{ if } $g(\tilde{\mathbf{\vartheta}}^{(o+1)}) \leq c^{(o)}-\zeta \|\tilde{\mathbf{\vartheta}}^{(o+1)}-\overline{\mathbf{\vartheta}}^{(o)} \|^2$ \textbf{then}
\STATE $\mathbf{\vartheta}^{(o+1)}=\tilde{\mathbf{\vartheta}}^{(o+1)}$
\STATE \textbf{else}
\STATE update $\hat{\mathbf{\vartheta}}^{(o+1)}$ \text { using \eqref{eq:vnhat+1} } and then $\mathbf{\vartheta}^{(o+1)}$ \text { using \eqref{eq:vn+1} }
\STATE \textbf{end if}
\STATE update $q^{(o+1)}$ using~\eqref{eq:q_n}.
\STATE update $b^{(o+1)}$ \text { using \eqref{eq:bn+1}  and } $c^{(o+1)}$ \text { using \eqref{eq:cn+1} }
\STATE  update $o=o+1$
\STATE \textbf{end while} 
\STATE \textbf{until}{ Convergence.}
\end{algorithmic}
\end{algorithm}
\vspace{-0.5em}
\section{Numerical Results}
Here, we present numerical results to evaluate the performance of the CF-mMIMO system with joint unicast and multi-group multicast transmissions, utilizing the proposed joint power control and AP selection approach based on APG.
We assume that there are $N$ APs, each equipped with  $L = 4$ antennas, to simultaneously  serve $U$ unicast users and $M$ multicast groups, while all the users and APs are  randomly distributed within an area of size $1 \times 1$ km$^2$.    The pilot length is $\tau=U+M$, while the bandwidth is $20$ MHz and $K_\mathrm{max}=U+M$. The large-scale fading coefficients $\beta_{n,u}$ and $\lambda_{n,m,k}$ are modeled as  \cite{Björnson:TWC:2020}
$\beta_{n,u} = 10^{\frac{\text{PL}_{n,u}^d}{10}}10^{\frac{F_{n,u}}{10}}$ and
$\lambda_{n,m,k} = 10^{\frac{\text{PL}_{n,k_m}^d}{10}}10^{\frac{F_{n,k_m}}{10}}$, respectively, where $10^{\frac{\text{PL}_{n,u}^d}{10}}$ and $10^{\frac{\text{PL}_{n,k_m}^d}{10}}$ are the path loss, $10^{\frac{F_{n,u}}{10}}$ and $10^{\frac{F_{n,k_m}}{10}}$ denote the shadowing effect with $F_{n,u}\in\mathcal{N}(0,4^2)$ and $F_{n,k_m}\in\mathcal{N}(0,4^2)$ (in dB) for unicast and multicast users, respectively; $\text{PL}_{n,u}^d$ and $\text{PL}_{n,k_m}^d$ are in dB and can be calculated as
$\text{PL}_{n,u}^d = -30.5-36.7\log_{10}\Big(\frac{d_{n,u}}{1\,\text{m}}\Big)$
and
$\text{PL}_{n,k_m}^d = -30.5-36.7\log_{10}\Big(\frac{d_{n,k_m}}{1\,\text{m}}\Big),$
respectively.
The correlation among the shadowing terms from the $n$-th AP to different $h \in \{\mathcal{U} \cup \mathcal{M}\}$  unicast and multicast users is given by  
\begin{align}
	 \mathbb{E}\{F_{n,h}F_{j,h'}\} \triangleq \begin{cases} 4^22^{-\omega_{h,h'}/9\,\text{m}}, & j = n, \\ 0, & \text{otherwise}, \end{cases}
\end{align}
where $\omega_{h,h'}$ is the physical distance between users $h$ and $h'$. 
The maximum transmission power for each AP is $1$ W, and for each user  is $100$ mW, while the noise power is $-92$ dBm.
\begin{figure}
			\centering 
\includegraphics[width=0.41\textwidth]{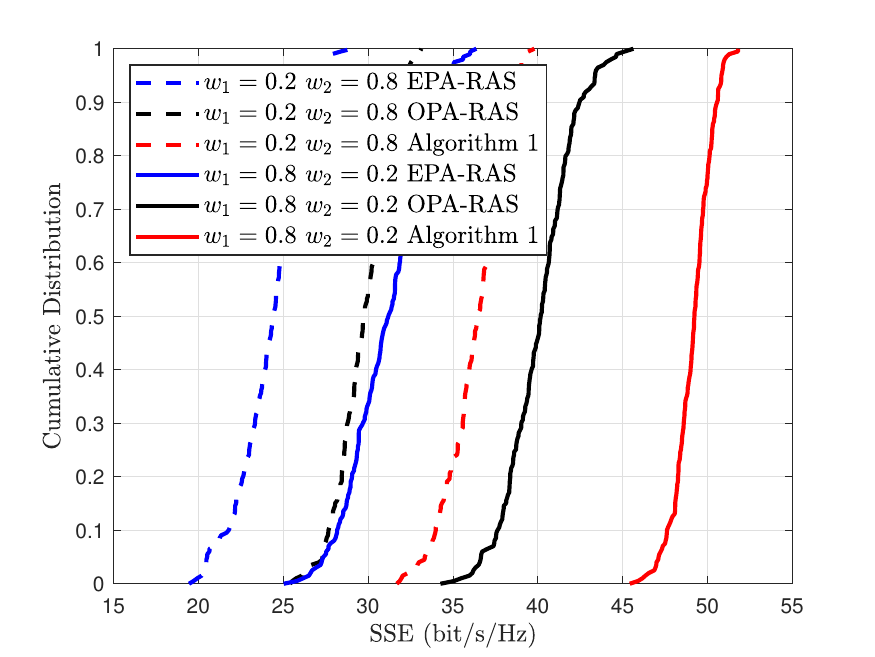}
\vspace{-0.2em}
   \caption{CDF of the SSE where the total number of unicast and multicast users is $28$ ($U=16$, $M=3$, $K_m=4$, $N=100$ and $\bar{SE}_{QoS}=SE_{QoS}=0.5$ bit/s/Hz).}
    \label{fig1}
\end{figure}
In Fig. \ref{fig1}, we examine the effectiveness of the APG-based joint power allocation and AP selection method outlined in \textbf{Algorithm 1}. To this end, the cumulative distribution function (CDF) of the SSE  of our  optimized approach in \textbf{Algorithm 1} is compared to those provided by  two benchmarks: i) equal power allocation (EPA) and random AP selection (RAS) and    ii) optimum power allocation (OPA) and  RAS selection. 
It is observed that with weighting coefficients of $w_1=0.8$ and $w_2=0.2$, the proposed joint optimization method yields a notable enhancement in the median SSE, showing improvements of approximately $58\%$ and $22\%$ compared to scenarios using random AP selection with EPA and OPA, respectively. These findings underscore the superiority of the proposed joint optimization approach over heuristic methods.
In addition, the performance improvement of our optimized approach grows as $w_1$ increases. 
This results from the enhanced flexibility provided by our scheme through the power control coefficients for unicast users, which improves interference management.

Figure \ref{fig2} shows the impact of the number of APs on the average SSE of  CF-mMIMO  system with joint unicast and
 multicast transmissions relying on our  APG-based   optimized approach. We consider three cases with different number of  multicast users. 
We can see that the SSE gain  improves significantly when the number of APs increases. Interestingly, the performance improvement is more significant in more challenging scenarios, such as when there is a higher number of multicast users in the network.

 
\begin{figure}
			\centering 
\includegraphics[width=0.41\textwidth]{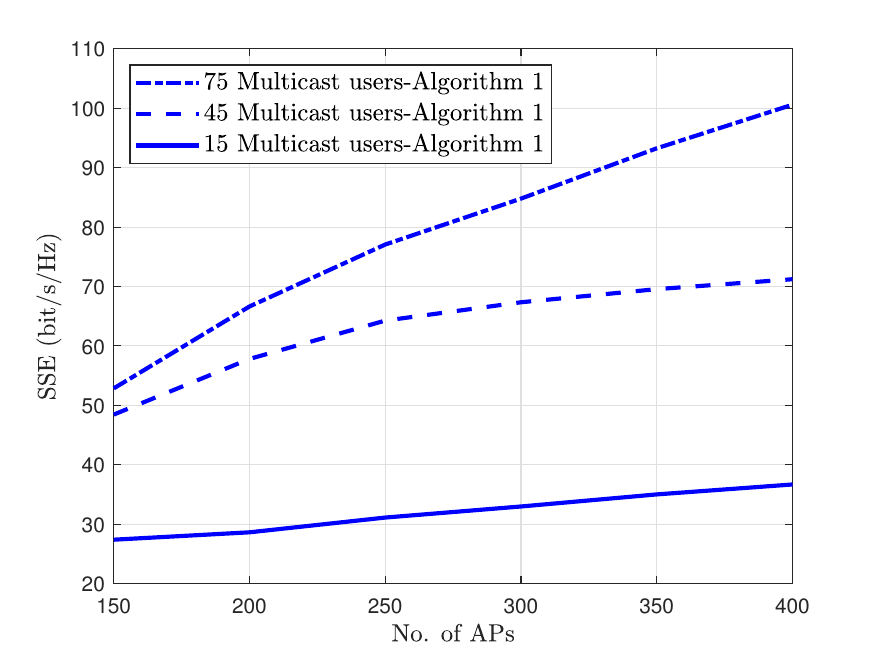}
\vspace{-0.2em}
   \caption{Average SSE versus the number of APs for different number of multicast users ($U=5$ and $\bar{SE}_{QoS}=SE_{QoS}=0.4$ bit/s/Hz).}
   \label{fig2}
\end{figure}
\vspace{-0.5em}
\section{Conclusion}
We have investigated the SSE  performance of a  CF-mMIMO  system with joint unicast and multi-group multicast transmissions and proposed a large-scale fading-based joint power allocation coefficient and user association optimization approach to maximize the SSE, subject to per-AP transmit power constraint  and  QoS  SE requirements for both unicast and multicast users. Our results indicated that the jointly optimized APG-based approach achieves a notably higher SSE compared to the benchmark methods, particularly in large-scale systems with higher numbers of multicast users in the network.  Future work will include a performance and complexity comparison between  SCA and APG-based optimization designs for CF-mMIMO systems with joint unicast and multicast transmissions, incorporating more advanced precoding schemes. 
\vspace{-0.5em}
\bibliographystyle{IEEEtran}
\bibliography{bibliography}
\end{document}